\documentclass[letterpaper, 10 pt, conference,onecolumn]{ieeeconf}
\usepackage{amsmath,amsfonts,amssymb}
\usepackage{graphicx,epsfig,psfrag,subfigure,remark}
\usepackage{xfrac,bm}
\usepackage{algorithmic,algorithm}
\usepackage[all]{xy}
\usepackage{varioref}
\usepackage{wrapfig}
\usepackage{threeparttable}
\usepackage{dcolumn}
\newcolumntype{d}{D{.}{.}{-1}}
\usepackage{nomencl}
\makeglossary
\usepackage{subfigure}
\usepackage{comment,cite}
\usepackage{xfrac}
\usepackage{mathrsfs}
\IEEEoverridecommandlockouts 

\newcommand{\ud}{\mathrm{d}}
\newcommand{\x}{{\bm{x}}}
\newcommand{\X}{{\mathsf{X}}}

\newcommand{\f}{{\mathsf{f}}}

\newcommand{\z}{{\bm{z}}}

\newcommand{\cA}{\mathcal{A}}

\newcommand{\cG}{\mathcal{G}}

\newcommand{\cI}{\mathcal{I}}
\newcommand{\cJ}{\mathcal{J}}

\newcommand{\cL}{\mathcal{L}}

\newcommand{\cP}{\mathcal{P}}

\newcommand{\cR}{\mathcal{R}}
\newcommand{\cS}{\mathcal{S}}

\newcommand{\cT}{\mathcal{T}}
\newcommand{\cU}{\mathcal{U}}

\newcommand{\mT}{T}

\newtheorem{proposition}{Proposition}

\newtheorem{definition}{Definition}
\newtheorem{problem}{Problem}

\renewcommand{\t}{^{\mbox{\scriptsize \mT}}}

\renewcommand{\baselinestretch}{1.05}
\newremark{remark}{Remark}

\title{\textbf{Decentralized Game-Theoretic Control for Dynamic Task Allocation Problems for Multi-Agent Systems}}

\author{Efstathios Bakolas \and Yoonjae Lee \thanks{This work was supported in part by ARL under W911NF2020085. E. Bakolas (Associate Professor) and Y. Lee (graduate student) are with the Department of Aerospace Engineering
and Engineering Mechanics, The University of Texas at Austin,
Austin, Texas 78712-1221, USA, Emails: bakolas@austin.utexas.edu; yol033@utexas.edu}}

\begin{document}

\maketitle

\begin{abstract}
We propose a decentralized game-theoretic framework for dynamic task allocation problems for multi-agent systems. In our problem formulation, the agents' utilities depend on both the rewards and the costs associated with the successful completion of the tasks assigned to them. The rewards reflect how likely is for the agents to accomplish their assigned tasks whereas the costs reflect the effort needed to complete these tasks (this effort is determined by the solution of corresponding optimal control problems). The task allocation problem considered herein corresponds to a dynamic game whose solution depends on the states of the agents in contrast with classic static (or single-act) game formulations. We propose a greedy solution approach in which the agents negotiate with each other to find a mutually agreeable (or individually rational) task assignment profile based on evaluations of the task utilities that reflect their current states. We illustrate the main ideas of this work by means of extensive numerical simulations.
\end{abstract}


\section{Introduction}\label{s:intro}

We consider a dynamic task allocation problem for a multi-agent system whose agents have continuous state and input spaces and have to complete a set of spatially distributed tasks (obtain in-situ measurements or pick up packages from different locations over a given spatial domain). We adopt a game-theoretic approach which seeks for task assignments that maximize the individual utility of each agent conditional on the assignments of their teammates (individual rationality principle) while also ensuring that the self-interests of the agents are aligned with those of the team. To this aim, we design the agents' utilities in accordance with the concept of \textit{wonderful life utility}~\cite{p:wolpert1999} (WLU) which allows us to associate the dynamic task allocation problem with a sequence of potential games~\cite{p:monderer1996potential}. We propose a greedy decentralized algorithm for the computation of task assignment profiles which are mutually agreeable in the long run.

\noindent \textit{Literature review:} Task allocation problems for multi-agent systems can be addressed by auction based techniques, distributed and / or multi-objective optimization and game-theoretic methods. The auction-based techniques are centralized when the agents negotiate with each other under the guidance of an auctioneer~\cite{p:gerkey2002} and decentralized when they negotiate directly with each other~\cite{p:howauction2009,p:NANJANATH2010,p:capitan2013}. Centralized methods rely on a single point of failure whereas the communication cost in decentralized methods can be substantial if not prohibitive. Distributed optimization~\cite{p:distrmaxsum} 
and multi-objective optimization~\cite{p:tolmidis2013} for task allocation problems are typically quite complex and require some knowledge about the utilities of the other agents and, more importantly, do not necessarily yield solutions which are mutually agreeable. As suggested in \cite{p:shamma2007}, game-theoretic tools constitute one of the most natural approaches to task allocation problems for intelligent, autonomous agents. Reference~\cite{p:shamma2007}, which is the main inspiration of this paper, utilizes the framework of potential games to define in a systematic way the task and agent utilities as well as several negotiation protocols (game-theoretic learning algorithms~\cite{p:fudenberg1998,b:young2020}) for the computation of mutually agreeable task assignments in a decentralized or distributed way. The negotiation protocols utilized in \cite{p:shamma2007} converge to mutually agreeable task assignment profiles without requiring that any agent should know the utility functions of her teammates (decentralized task allocation). However, their convergence is conditional on the game remaining the same (e.g., the functional description of the utilities does not change throughout the negotiation process). Thus, although the equilibrium of the game is found iteratively in \cite{p:shamma2007}, the task allocation problem itself is essentially modeled as a static game. Extensions of the game-theoretic framework for multi-agent control problems can be found in \cite{p:marden2009,p:chasparis2011,p:qpotential2014}. 
Ref. \cite{p:chapman2010} proposes a myopic solution approach to a dynamic task allocation problem modeled as a sequence of static (single-act) potential games. The approach in \cite{p:chapman2010} cannot handle state-dependent utilities in general. Finally, the framework of state-dependent potential games~\cite{p:MARDEN2012} is only applicable to problems with finite (discrete) state spaces.


\textit{Contributions:} In this paper, we address a dynamic task allocation problem in which the task utilities depend on both the rewards earned by the agents for completing their assigned tasks as well as the costs they incur while doing so (cost-to-go functions of corresponding optimal control problems). Consequently, the utilities are in general state-dependent. We adopt a decentralized game-theoretic solution approach (each agent knows only her own utility function). The (individual) agent utilities are designed in accordance with the WLU framework which ensures that their self-interests are aligned with the team's interests under the framework of potential games. We propose a greedy solution approach in which the negotiations between the agents take place on-the-fly while the agents move in their state space towards their assigned tasks. Every time an agent changes her individual assignment (and thus her final state destination) she has to update the estimate of the cost-to-go and consequently her utility function as well (state-dependent utilities). We design the negotiation process such that the agents compute a mutually agreeable profile which is not likely to change during the last phase of the process. 

\textit{Outline:} The rest of the paper is organized as follows. In Section~\ref{s:prelim}, we discuss the problem preliminaries. The task, team and agent utilities are defined in Section~\ref{s:utility}. The open-loop task allocation is addressed in Section~\ref{s:OLTA} and the dynamic problem in Section~\ref{s:RHTA}. Numerical simulations are presented in Section~\ref{s:simu}. Finally, Section~\ref{s:concl} presents concluding remarks and directions for future work.

\section{Preliminaries and Problem Setup}\label{s:prelim}

\textit{Notation:} We denote by $\mathbb{R}^n$ the set of $n$-dimensional real vectors. We denote by $\mathbb{Z}$ the set of integers. Given $a,b \in \mathbb{Z}$ with $a \leq b$, we denote by $[a,b]_d$ the discrete time interval from $a$ to $b$, that is, $[a,b]_d := [a,b]  \cap  \mathbb{Z}$. Given $k \in \mathbb{Z}$, we write $\mathbb{Z}_k$ to denote the (unbounded) discrete interval $[k,\infty) \cap \mathbb{Z}$. Given a finite set $\cA$, we denote by $\mathrm{card}(\cA)$ its cardinality.


\textit{Problem setup:} We consider a multi-agent system (MAS) comprised of $n$ agents. We denote by $x_i \in \cS_i \in \Sigma$ and $u_i \in \cU_i$, for $i \in [1,n]_d$, the state and input of the $i$-th agent of the MAS at time $t\geq 0$, where $\cS_i$ and $\cU_i$ denote her state space and input space, respectively, and $\Sigma \subseteq \mathbb{R}^m$. In addition, we denote by $\bm{x} \in \cS$ the joint state of the MAS, where $\bm{x} := (x_1,\dots,x_n)$ and $\cS:= \cS_1 \times \dots \times \cS_n$ (joint state space), and by $\bm{u} \in \cU$ the joint input of the MAS, where $\bm{u} := (u_1, \dots, u_n)$ and $\cU:= \cU_1 \times \dots \times \cU_n$ (joint input space). Furthermore, we denote by $\bm{x}_{-i} \in \cS_{-i}$ and $\bm{u}_{-i} \in \cU_{-i}$ the concatenations of the states and the inputs of all the agents except from the $i$-th agent (the sets $\cS_{-i}$ and $\cU_{-i}$ are defined accordingly).

The motion of the $i$-th agent is described by
\begin{equation}\label{eq:motion1}
\dot{x}_i = f_i(\bm{x}, \bm{u}),~\quad~x_i(0)=x_i^0,~\quad~i \in [1, n]_d,
\end{equation}
where $x^0_i \in \cS_i$ is the initial state of the $i$-th agent and $f_i: \cS_i \times \cU_i \rightarrow \cS_i$ is her associated vector field. Note that the evolution of the $i$-th agent is not fully determined by her own state and input. For instance, in any realistic setting, the input of every agent at each time is conditioned on the actions of the other agents or at least a subset of them. A similar argument can be made for the states of the other agents. We assume that the vector field $f_i$ satisfies regularity conditions that ensure the existence and uniqueness of solutions to the differential equations \eqref{eq:motion1} for all piecewise continuous joint inputs $\bm{u}$ taking values in $\cU$ and all joint states $\bm{x} \in \cS$. Finally, we write
\begin{equation}\label{eq:motionT}
\dot{\bm{x}} = \bm{f}(\bm{x}, \bm{u}),~\quad~ \bm{x}(0)=\bm{x}^0,
\end{equation}
where $\bm{x}^0 = (x^0_1, \dots, x^0_n) \in \cS$ is the joint initial state and $\bm{f}:=(f_1, \dots, f_n)$ is the joint vector field.

The task allocation problem seeks for individual assignments for a team of $n$ agents and for a given set of $p$ tasks, $\cT := \{ \cT_1, \dots, \cT_p\}$. Each task is associated with a distinct state in $\Sigma$. We denote by $X_{\cT}$ the set of states associated with the given tasks, where $X_{\cT} := \{ x_{\cT_1}, \dots, x_{\cT_p} \}$. In principle, an agent can be assigned at most one task in $\cT$ at each instant of time although more than one agents can be assigned to the same task simultaneously. We denote by $\cA_i:= \{ a_i^k:~k\in[1,\mathrm{card}( \cA_i )]_d \}$ the set of possible task assignments for the $i$-th agent given a set of tasks $\cT$. Later on, we will see that each assignment $a_i^k$ induces a corresponding (admissible) control input $u_i(\cdot)$ via the solution of a corresponding optimal control problem. We assume that 
$a_{\varnothing} \in \cA_i$, where $a_{\varnothing}$ denotes the null assignment (i.e., the $i$-th agent is not assigned to any task) which corresponds to the null control input, that is, when $a_i = a_{\varnothing}$, then $u_i(t) = 0$, for all $t \geq 0$. Each assignment $a_i^k \in \cA_i$ is equal to either a task in $\cT$, that is, $a_i^k = \cT_\ell$ where $\cT_\ell \in \cT$, or the null assignment, that is, $a_i^k = a_\varnothing$. Thus, $\cA_i \subseteq \overline{\cT}$, where $\overline{\cT} := \cT \cup 
\{ a_{\varnothing} \}$.

\section{Task Utilities}\label{s:utility}
The completion of a task $\cT_j \in \cT$ will accrue rewards to the agents assigned to it. These rewards, which do not depend on the states of the agents, reflect the  importance of this specific task as well as the likelihood of its successful completion by each agent assigned to it (in general, not all agents are equally likely to complete a specific task successfully). We will refer to these rewards as the static task utilities. Furthermore, an agent will have to incur a cost to complete her assigned task (e.g., the transition cost to a certain location associated with this task). It is worth mentioning that the task completion cost is state-dependent and we will refer to it as the \textit{dynamic task completion cost}. 

\noindent \textit{Static task utility:} Given an action profile $\bm{a} = (a_1, \dots, a_n)$, we denote by $\cT^{-1}_j( \bm{a} )$ the index-set corresponding to the agents assigned to task $\cT_j \in \cT$ under the particular profile, that is, $\cT^{-1}_j( \bm{a} ) =\{i \in [1,n]_d: a_i = \cT_j\}$. The completion of task $\cT_j$ will accrue a reward $r_{\cT_j} \geq 0$ to the agent or agents assigned to this task. In general, $r_{\cT_j}$ is a function of the task assignment profile, that is,
\begin{equation}
r_{\cT_j(\bm{a})} = \bar{r}_{\cT_j}[1-\prod\nolimits_{i \in \mathcal{T}_j^{-1} (\bm{a})} (1-p_{ij})], \label{eq}
\end{equation}
where $\bar{r}_{\cT_j}$ is the nominal reward of $\cT_j$ and $p_{ij} \in [0,1]$ is the probability of the task $\cT_j$ be completed successfully by the $i$-th agent. If 
$\cT^{-1}_j( \bm{a} ) = \varnothing$, then $r_{\cT_j}( \bm{a} ) := 0$.

\noindent \textit{State-dependent task completion cost:} Next, we define the cost for completing the task $\cT_j$ associated with the state $x_{\cT_j}$ at time $t=t_\f$ by the $i$-th agent. Essentially, the task completion cost is taken to be the cost incurred by the $i$-th agent, which starts from the state $x_i^0$ at time $t=0$, to reach the state $x_{\cT_j}$ at time $t=t_\f$. The latter state transition cost is defined as the optimal cost-to-go corresponding to the following optimal control problem:
\begin{problem}\label{problemOCP}
Let $a_i = \cT_j$, where $\cT_j \in \cT$ and $i \in [1,n]_d$. Furthermore, let $x_{\cT_j} \in \cS_i$ be the state associated with the task $\cT_j$ and let $t_\f>0$ be the corresponding  completion time (fixed and common for all the agents). Then, find an optimal piece-wise continuous input $u_i^{\star}(\cdot): [0,t_\f] \rightarrow \cU_i$ that minimizes the following performance index:
\begin{align}\label{eq:cost}
\cJ_i(u_i(\cdot); x_i^0, x_{\cT_j}) 
%
%
:= \int_0^{t_\f} \cL_i(x_i(t), u_i(t)) \ud t,
\end{align}
subject to the dynamic constraints \eqref{eq:motion1} and the terminal constraint: $\Psi_i( x_i(t_\f), x_{\cT_j} ) = 0$, where $\Psi_i(\cdot; x_{\cT_j})$ is a given $C^1$ function. Finally, the optimal cost-to-go is denoted by $\rho_i(x_i^0; x_{\cT_j})$, where 
$\rho_i(x_i^0; x_{\cT_j}) := \cJ_i(u_i^{\star}(\cdot);x_i^0, x_{\cT_j})$. 
\end{problem}
\begin{remark}
%
%
The terminal constraint function $\Psi_i$ can be defined, for instance, as follows: $\Psi_i( x_i(t_\f), x_{\cT_j} ) = x_i(t_\f) - x_{\cT_j}$, in which case we require that $x_i(t_\f) = x_{\cT_j}$ (hard constraint).
\end{remark}

\noindent \textit{Total Task Utility:} The total cost of completion of task $\cT_j$ under the action profile $\bm{a} = (a_1, \dots, a_n)$, which is denoted as $\cR_{\cT_j}(\bm{a};\bm{x}^0, x_{\cT_j})$, is defined as the sum of the individual task completion costs of all the agents assigned to that task. More precisely,
\begin{equation}
\cR_{\cT_j}(\bm{a};\bm{x}^0, x_{\cT_j}) := \sum\nolimits_{i \in \cT^{-1}_j(\bm{a}) } \rho_i(x_i^0; x_{\cT_j}).
\end{equation}
Note that $\cR_{\cT_j}$ depends on the initial state $ \bm{x}^0$ (more precisely, the initial states of the agents assigned to the task $\cT_j$).
Furthermore, the total task utility associated with task $\cT_j$ for a given $\bm{x}_0$ is denoted as $\cU_{\cT_j}(\bm{a};\bm{x}^0)$ and defined as follows:
\begin{equation}\label{eq:teamutility}
\cU_{\cT_j}(\bm{a};\bm{x}^0) := \max\{0, r_{\cT_j}(\bm{a}) - \cR_{\cT_j}(\bm{a},\bm{x}^0; x_{\cT_j})\}.
\end{equation}
Note that $\cU_{\cT_j}$ is state-dependent because the task completion costs $\rho_i(x_i^0; x_{\cT_j})$, for $i \in \cT^{-1}_j(\bm{a})$, are state-dependent.

\noindent \textit{Individual and Team Utilities and Solution Concepts:} 
%
First, we define the team's utility (the latter reflects the team's collective welfare), which is denoted by $\cU(\bm{a};\bm{x}^0)$, as follows:
\begin{equation}\label{eq:teamutility}
\cU(\bm{a};\bm{x}^0) := \sum\nolimits_{\cT_j\in \cT} \cU_{\cT_j}(\bm{a};\bm{x}^0).
\end{equation}

The individual utility of the $i$-th agent given a task profile $\bm{a} = (a_i, a_{-i})$, which is denoted as $\cU_{i}(a_i, a_{-i})$ or $\cU_{i}(\bm{a})$, is taken to be equal to her marginal contribution to the team's utility $\cU(\bm{a};\bm{x}^0)$, that is,
\begin{equation}\label{eq:marginaleq1}
\cU_{i}(\bm{a};\bm{x}^0) := \cU((a_i,a_{-i}); \bm{x}^0) - \cU((a_{\varnothing}, a_{-i}); \bm{x}^0)
\end{equation}
from which it follows, in view of \eqref{eq:teamutility}, that
\begin{equation}\label{eq:marginaleq2}
\cU_{i}(\bm{a};\bm{x}^0) = \cU_{\cT_j}((a_i,a_{-i}); \bm{x}^0) - \cU_{\cT_j}((a_{\varnothing}, a_{-i}); \bm{x}^0),
\end{equation}
where $(a_{\varnothing}, a_{-i})$ corresponds to the action profile when the $i$-th agent has a null assignment, that is, $a_i = a_{\varnothing}$. 
Next, we provide the definition of the basic solution concept that will be used in our task allocation problem.

\begin{definition} 
An assignment profile $\bm{a}^{\star}:=(a^{\star}_i, a^{\star}_{-i}) \in \cA$ is a pure strategy Nash equilibrium of the game $\cG^0$, where $\cG^0 := \langle \cU_{1}(\bm{a}; \bm{x}^0), \dots, \cU_{n}(\bm{a}; \bm{x}^0) ; \cA \rangle$,
if $\forall i \in [1,n]_d$
\begin{equation}\label{eq:agentutility}
\cU_{i}(a^{\star}_i, a^{\star}_{-i}; \bm{x}^0) \geq \cU_{i}(a_i, a^{\star}_{-i}; \bm{x}^0),~~~~\forall a_i \in \cA_i.
\end{equation}
\end{definition}

\begin{remark} The solution concept of (pure strategy) Nash equilibrium is fundamental in non-cooperative game theory. When all agents play in accordance with the Nash equilibrium, they act selfishly and try to maximize their own utilities conditional on the decisions of others (\textit{individual rationality}).
\end{remark}

\section{The Open-Loop Task Allocation Problem and its Decentralized Solution}\label{s:OLTA}

\subsection{Problem formulation and analysis}
Next, we formulate the task allocation problem as a non-cooperative game. In the following formulation, we only account for the estimates of the task completion costs at the initial time (open-loop approach).

\begin{problem}[OLTA: Open-Loop Task Allocation]\label{problemOLTA}
Let $t_\f > 0$ and $\bm{x}^0 \in \cS$ be given. Then, find a (time-invariant) task assignment profile $\bm{a}^{\star} \in \cA$, where $\bm{a}^{\star} :=(a^{\star}_i, a^{\star}_{-i})$, such that for all $i \in [1,n]_d$ the inequality \eqref{eq:agentutility} is satisfied.
In other words, the task assignment profile $\bm{a}:=(a^{\star}_i, a^{\star}_{-i})$ corresponds to a Nash equilibrium of the game $\cG^0$.
\end{problem}

\begin{remark}
In the formulation of the open-loop task allocation problem (Problem~\ref{problemOLTA}), the agents' utilities (or, more precisely, their functional descriptions) do not change with time, as the agents progress towards the states of their assigned tasks. This is because their estimated task completion costs are based on knowledge available at time $t=0$ and these estimates are not updated afterwards. To the $i$-th individual task assignment $a^\star_i$ from the optimal profile $\bm{a}^{\star}$, where, say, $a^\star_i = \cT_j$, we associate a corresponding state $x_{\cT_j}$, which in turn determines the terminal constraint $\Psi(x_i(t_\f);x_{\cT_j})=0$ in Problem~\ref{problemOCP}. Because all the task assignments are time-invariant, the control input $u_i^{\star}(\cdot)$ that solves Problem~\ref{problemOCP} will not be updated along the $i$-th agent' ensuing trajectory.  
\end{remark}

It is well-known that potential games correspond to a special class of  non-cooperative games that always admit pure strategy Nash equilibria. We claim that Problem~\ref{problemOLTA} corresponds to an exact potential game~\cite{p:potgames2016}.

\begin{definition}\label{def:potential}
The game, $\cG^0$, 
%
%
corresponds to an exact potential game, if there exists an exact potential, that is, a function $\cP: \cA \rightarrow \mathbb{R}$ such that, for all $i\in [1, n]_d$ and $\bm{a} \in \cA$, it holds true that
\begin{align}\label{eq:potential}
& \cU_{i}(a'_i, a_{-i}; \bm{x}^0) - \cU_{i}(a_i, a_{-i}; \bm{x}^0) \nonumber \\
& ~~~\qquad \qquad ~~~ = \cP(a'_i, a_{-i}) - \cP(a_i, a_{-i}),
\end{align}
for all $a_i, a_i' \in \cA_i$.
\end{definition}

\begin{proposition}
The open-loop task allocation problem (Problem~\ref{problemOLTA}) with task team utility and individual utilities defined by \eqref{eq:teamutility} and \eqref{eq:agentutility}, respectively,
is an exact potential game with potential $\cP = \cU(a_i, a_{-i}; \bm{x}^0)$, for all $a_i \in \cA_i$.
\end{proposition}
\begin{proof}
We will directly verify \eqref{eq:potential}. In particular,
\begin{align*}
 & \cU(a'_i, a_{-i};\bm{x}^0) - \cU(a_i, a_{-i};\bm{x}^0)
\\ & ~~~ = \sum_{\cT_j \in \cT} \big( \cU_{\cT_j}(a'_i, a_{-i};\bm{x}^0) - \cU_{\cT_j}(a_i, a_{-i};\bm{x}^0) \big)
\\ & ~~~ = \sum_{\substack{k=1\\ k\neq i}}^{p}  \cU_{\cT_j}(a_k, a_{-k};\bm{x}^0)  - \sum_{k=1}^{p}  \cU_{\cT_j}(a_k, a_{-k};\bm{x}^0) 
\\ &~~~~~~ + \cU_{\cT_j}(a'_i, a_{-i};\bm{x}^0) 
\\ &~~~ = \cU_{\cT_j}(a'_i, a_{-i};\bm{x}^0) - \cU_{\cT_j}(a_i, a_{-i};\bm{x}^0)
\\ &~~~ = \big( \cU_{\cT_j}(a'_i, a_{-i};\bm{x}^0) - \cU_{\cT_j}(a_\varnothing,a_{-i};\bm{x}^0) \big)
\\ &~~~~~~ - \big( \cU_{\cT_j}(a_i, a_{-i};\bm{x}^0) - \cU_{\cT_j}(a_\varnothing,a_{-i};\bm{x}^0) \big)
\\ & ~~~ = \cU_i(a'_i,a_{-i};\bm{x}^0) - \cU_i(a_i,a_{-i};\bm{x}^0),
\end{align*}
where in the derivation of the last equality we have used Eq.~\eqref{eq:marginaleq2}. In view of Definition~\ref{def:potential}, we conclude that $\cG^0$ corresponds to an exact potential game.
\end{proof}


\subsection{Negotiation protocols for decentralized task allocation}\label{ss:negotiate}
Problem~\ref{problemOLTA} can be solved by utilizing standard tools for the computation of Nash equilibria of noncooperative (static or single-act) games and in particular potential games~\cite{p:potgames2016}. An alternative approach is to employ game-theoretic learning algorithms which generate a sequence of task assignment profiles that converge to a Nash equilibrium.
Some of these algorithms include the fictitious play (FP), spatial adaptive play (SAP), and generalized regret matching (GRM) algorithms to name but a few; the reader may refer to~\cite{p:shamma2007} for more information on these and other similar algorithms. A key point is that for their realization, an agent does not have to know the utilities of her teammates (decentralized implementation). 

During the negotiation (learning) process, the task assignment profile of the team is updated at different time instants that form a non-decreasing sequence $\{ \tau_k\}_{k=0}^\infty$ in $[0,t_\f]$ such that $\tau_0=0$ and $\lim_{k \rightarrow \infty} \tau_k = t_\f$. The time instant $\tau_k$ corresponds to the $k$-th stage of the negotiation process. At that stage, the $i$-th agent picks her new task assignment, which we denote as $a_i(\tau_k)$; we also denote the corresponding profile of the whole team as $\bm{a}(\tau_k)$. The exact definition of the $a_i(\tau_k)$ will be determined by the particular learning algorithm that will be employed, which in turn will rely on a corresponding information set $\cI^i_k$. The latter set may encode information about the past performance of the $i$-th agent (measured in terms of past values of her own utility) as well as information about the history (whole or truncated) of her teammates' actions (such information may correspond to, for instance, the empirical distribution of the agents' past task assignments). We write
\begin{equation}\label{eq:updatelawG0}
    a_i(\tau_k) = \phi_i(\cI^i_k; \cG^0),~~~ \bm{a}(\tau_k) = \bm{\phi}(\bm{\cI}_k; \cG^0),~~~i \in [1,n]_d,
\end{equation}
where $\phi_i: \cI^i_k \rightarrow \cA_i$, for $i \in [1, n]_d$, is the update law (or proposal) of the individual target assignment of the $i$-th agent whereas $\bm{\phi}: \bm{\cI}^k \rightarrow \cA$, with $\bm{\phi}(\bm{\cI}_k; \cG^0) := (\phi_1(\cI^1_k; \cG^0), \dots, \phi_n(\cI^n_k; \cG^0))$ is the update law of the task assignment profile of the MAS given the joint information set $\bm{\cI}_k := \cI^1_k \times \dots \times \cI^n_k$. The following claim is based on the analysis provided in \cite{p:shamma2007} (refer to, for instance, Theorem 4.1) and references therein.

\textit{Claim 1:} The update law $\bm{a}(\tau_k)$ which is defined as in \eqref{eq:updatelawG0} and corresponds to one of the decentralized negotiation protocols (game-theoretic learning algorithms) used \cite{p:shamma2007} will converge to a pure strategy Nash equilibrium of the game $\cG^0$.

It is important to emphasize that the update law for the task assignment profile will solve Problem~\ref{problemOLTA} under the assumption that throughout the interval $[0,t_\f]$, the functional description of the utilities will be based on their initial estimates at time $t=0$. However, the agents' utilities are state-dependent and thus their functional description will change along the agents' ensuing trajectories. This variability in the agents' preferences and capabilities (as reflected on their utilities) cannot be captured in this update law as well as the OLTA problem itself.
An alternative interpretation of the negotiation process is to assume that it does not take place over the time interval $[0,t_\f]$ but instantaneously, at time $t=0$. In other words, the clock is paused until the negotiations have converged (within some acceptable tolerance) to a Nash equilibrium of the potential game $\cG^0$. Subsequently, the agents can execute the corresponding inputs that will transfer them to the terminal states associated with their assigned tasks (these inputs are computed by solving Problem~\ref{problemOCP} for each agent). The input signal  will remain the same function of time, for all $t \in [0, t_\f]$. It is worth mentioning that the game-theoretic learning algorithms can be implemented based on local information (distributed implementation) by requiring that an agent cannot be assigned a task which is not within a certain range $\varrho >0$ from her (\textit{range constrained} case) in contrast with the nominal (\textit{range unconstrained} case) in which $\varrho \rightarrow \infty$.

\section{Dynamic Task Allocation and a Greedy Algorithm for its Solution}\label{s:RHTA}

\subsection{Problem formulation}
Next, we formulate a dynamic version of the task allocation problem in which the fact that the agents' utilities change along their ensuing trajectories is accounted in the determination of their task assignments in contrast with the OLTA problem. In this problem formulation, a new game $\cG^t$, where $\cG^t:=\langle \cU_{1}(\cdot; \bm{x}(t)), \dots, \cU_{n}(\cdot; \bm{x}(t)); \cA \rangle$ is essentially obtained at each $t \in [0,t_\f]$ as the agents move in their state space. 
%

\begin{problem}[DTA: Dynamic Task Allocation]\label{problemRHTA}
Let $t_\f > 0$ and $\bm{x}^0 \in \cS$ be given. Then, find a time-varying task assignment profile $\bm{a}^{\star}(\cdot): [0,t_\f] \rightarrow \cA$, where $\bm{a}^{\star}(t) :=(a^{\star}_i(t), a^{\star}_{-i}(t))$, which is such that for all $i \in [1,n]_d$:
\begin{equation*}
\cU_{i}(a^{\star}_i(t), a^{\star}_{-i}(t); \bm{x}(t)) \geq \cU_{i}(a_i(t), a^{\star}_{-i}(t); \bm{x}(t)),
\end{equation*}
for all $a_i(t) \in \cA_i$, as $t \rightarrow t_\f$. In other words, the task assignment profile $\bm{a}^{\star}(t)$ converges to a Nash equilibrium of the game $\cG^t$ as $t \rightarrow t_\f$.
\end{problem}
\begin{remark}
Problem~\ref{problemRHTA} seeks for a task assignment profile that will become mutually agreeable as $t$ approaches the final time $t_\f$. Note that if $t_\f$ is taken to be sufficiently large, then the solution to the Problem~\ref{problemRHTA} will essentially converge to a \textit{steady-state} task assignment profile.
\end{remark}

\subsection{A greedy algorithm for task allocation}\label{s:greedy}

Next we propose a greedy solution approach to address Problem~\ref{problemRHTA}. To ensure that the game-theoretic learning algorithms discussed in Section~\ref{ss:negotiate} will converge to a Nash equilibrium as $t\rightarrow t_\f$, we propose to stop updating the agents' utilities at time $t=t_\f-\epsilon$ for some $0<\epsilon<t_\f$ so that the utilized learning algorithm (whose convergence is guaranteed only for a static game) are given the chance to converge during the sub-interval $[t_\f-\epsilon, t_\f]$. A key difference between the DTA and OLTA problems is that the static game that determines the task assignment profile in the former is not $\cG^0$ (game corresponding to time $t=0$), as in the latter problem, but a game corresponding to time $t=t_\f - \epsilon$, assuming that the (dynamic) game has evolved for $t\in [0,t_\f - \epsilon]$.

Next, we present the main steps of the proposed greedy algorithm. To this aim, let us consider a sequence $\{\tau_k \}_{k=0}^{\infty}$ as in Section~\ref{ss:negotiate} and let $K = K(\epsilon)$ be the first positive integer at which $\tau_K> t_\f - \epsilon$ for the given $\epsilon$ (the existence of such $K$ is guaranteed by the fact that $\tau_k\rightarrow t_\f$ as $k\rightarrow \infty$). 
Now, let $\bm{\phi}(\bm{\cI}_k; \cG^{\tau_k})$ denote the update law of a negotiation protocol as in Section~\ref{ss:negotiate}. Let us consider the following update law:
\begin{align}\label{eq:updatedyn}
\bm{a}_d(\tau_k)& := \begin{cases} \bm{\phi}(\bm{\cI}_k; \cG^{\tau_k}), & \mathrm{for}~k \in [0,K-1]_d, \\
\bm{\phi}(\bm{\cI}_k; \cG^{\tau_K}), & \mathrm{for}~k \in \mathbb{Z}_{K}.
\end{cases}
\end{align}
We claim that the update law \eqref{eq:updatedyn} will find an approximate (in the sense that we will explain shortly later) solution to Problem~\ref{problemRHTA}. 

\begin{proposition}\label{prop:mainresult}
The piecewise constant dynamic task assignment profile
$\bm{a}(t) = \bm{a}_d(\tau_k),~~\forall t\in [\tau_k, \tau_{k+1})$
for all $k \in \mathbb{Z}_{\geq 0}$, where $\bm{a}_d$ is defined as in \eqref{eq:updatedyn}, will converge, as $t \rightarrow t_\f$, to a Nash equilibrium of the game $\cG^{T}$, where $T \in [t_\f -\epsilon, t_\f]$.
\end{proposition}  
\begin{proof}
Given that $\{\tau_k \}_{k=0}^{\infty}$ is a non-decreasing sequence in $[0,t_\f]$ which converges to $t_\f$ as $k \rightarrow \infty$, we conclude that $\tau_k \in [t_\f-\epsilon, t_\f]$, for all $k\in \mathbb{Z}_K$. After truncating the $K-1$ first elements of $\{\tau_k \}_{k=0}^{\infty}$, we obtain a new non-decreasing sequence $\{\sigma_n\}_{n=0}^{\infty}$, where $\sigma_n=\tau_{n+K}$, for $n \in \mathbb{Z}_{0}$, which implies that $\sigma_0=\tau_K$ and $\lim_{n\rightarrow \infty}\sigma_n= t_\f$. In view of Claim~1, the update law $\bm{a}(\tau_k) :=  \bm{\phi}(\bm{\cI}_k; \cG^0)$ defined in \eqref{eq:updatelawG0} will converge to a Nash equilibrium of the game $\cG^0$ as $k \rightarrow \infty$. From real analysis, we know that after truncating the first $K-1$ elements of the convergent sequence $\{ \bm{a}(\tau_k)\}_{k=0}^{\infty}$, we obtain a new sequence $\{ \bm{a}(\sigma_n)\}_{n=0}^{\infty}$ that will remain convergent with the same limit. The previous claim on convergence holds true for any game $\cG^{\tau_k}$ for a fixed $k$ when the latter is treated as a static game (with a possibly different limit for each $\tau_k$). We conclude that the sequence $\{ \bm{a}_d(\sigma_n)\}_{n=0}^{\infty}$, where $\bm{a}_d$ is defined in \eqref{eq:updatedyn}, will also converge to a Nash equilibrium of the game $\cG^{\tau_K}$, where by definition $\tau_{K} \in [t_\f - \epsilon, t_\f]$. This concludes the proof. 
\end{proof}
\begin{remark}
If at time $t=\tau_k$ the individual assignment of the $i$-agent attains a different value than at $t= \tau_{k-1}$, then the state corresponding to her new task will also be different. Therefore, the $i$-th agent will have to solve Problem~1 with the updated terminal constraint, with her initial state set equal to $x_i(\tau_k)$ and the final time to $t_\f - \tau_k$.
\end{remark}

\section{Numerical Simulations}\label{s:simu}

In this section, we present numerical simulations to illustrate the main ideas of the methods proposed so far. We consider a team of agents with double integrator dynamics, that is, $\ddot{p}_i = u_i$, 
%
%
with $p_i(0)=p_i^0$ and $\dot{p}_i(0)=v_i^0$, where $p_i \in \mathbb{R}^2$ ($p_i^0 \in \mathbb{R}^2$) and $\dot{p}_i \in \mathbb{R}^2$ ($v_i^0 \in \mathbb{R}^2$) denote, respectively, the position and velocity of the $i$-th agent at time $t$ ($t=0$), $i \in [1,n]_d$. The performance index is given by $
\cJ(u_i(\cdot)) := (1/2) \int_{0}^{t_\f} |u_i(t)|^2 \mathrm{d}t$ whereas the terminal constraint function $\Psi_i(\x_i(t_\f); x_{\cT_j}) := x_i - x_{\cT_j}$, where $x_i := (p_i,~\dot{p}_i) \in \mathbb{R}^4$ and $x_{\cT_j} := (p_{\cT_j}, 0) \in \mathbb{R}^4$, which means that the $i$-th agents tries to reach the position $p_{\cT_j}$ associated with her assigned task $\cT_j$ at time $t=t_\f$ with zero terminal velocity (soft landing).
It turns out (see, for instance, \cite{p:BakACC14a}) that the optimal control input is given by 
$u_i^{\star}(t;t_\f, x_i^0)  = \alpha + t \beta$, $\alpha := (6/t_\f^2) (p_{\cT_j} - p_i^0 - t_\f v^0_i) + (2/t_\f) v^0_i$, $\beta := -(12/ t_\f^3)( p_{\cT_j}- p_i^0 - t_\f v^0_i) - (6/ t_\f^2 ) v^0_i$,
and the optimal cost-to-go by
$\rho_i(x_i^0;x_{\cT_j}) := (1/2)(t_\f |\alpha|^2 + t_\f^2 \alpha^{\mathrm{T}} \beta + (1/3)t_\f^3 |\beta|^2)$.


\begin{figure}
\centering
\subfigure[$t=0^+$]{\includegraphics[width=5.2cm]{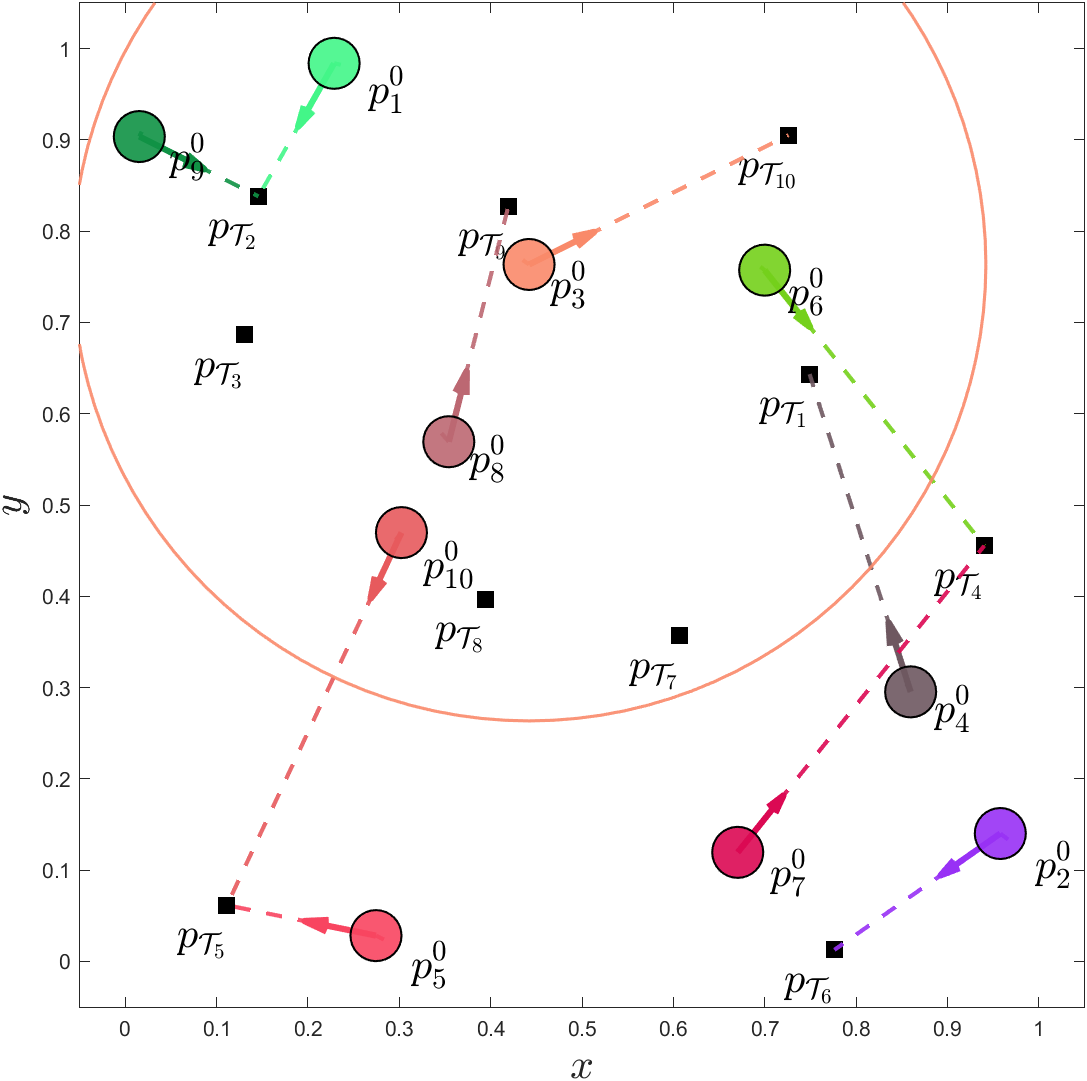}\label{fig:Ex_Im}}
\subfigure[$t=2.6$]{\includegraphics[width=5.2cm]{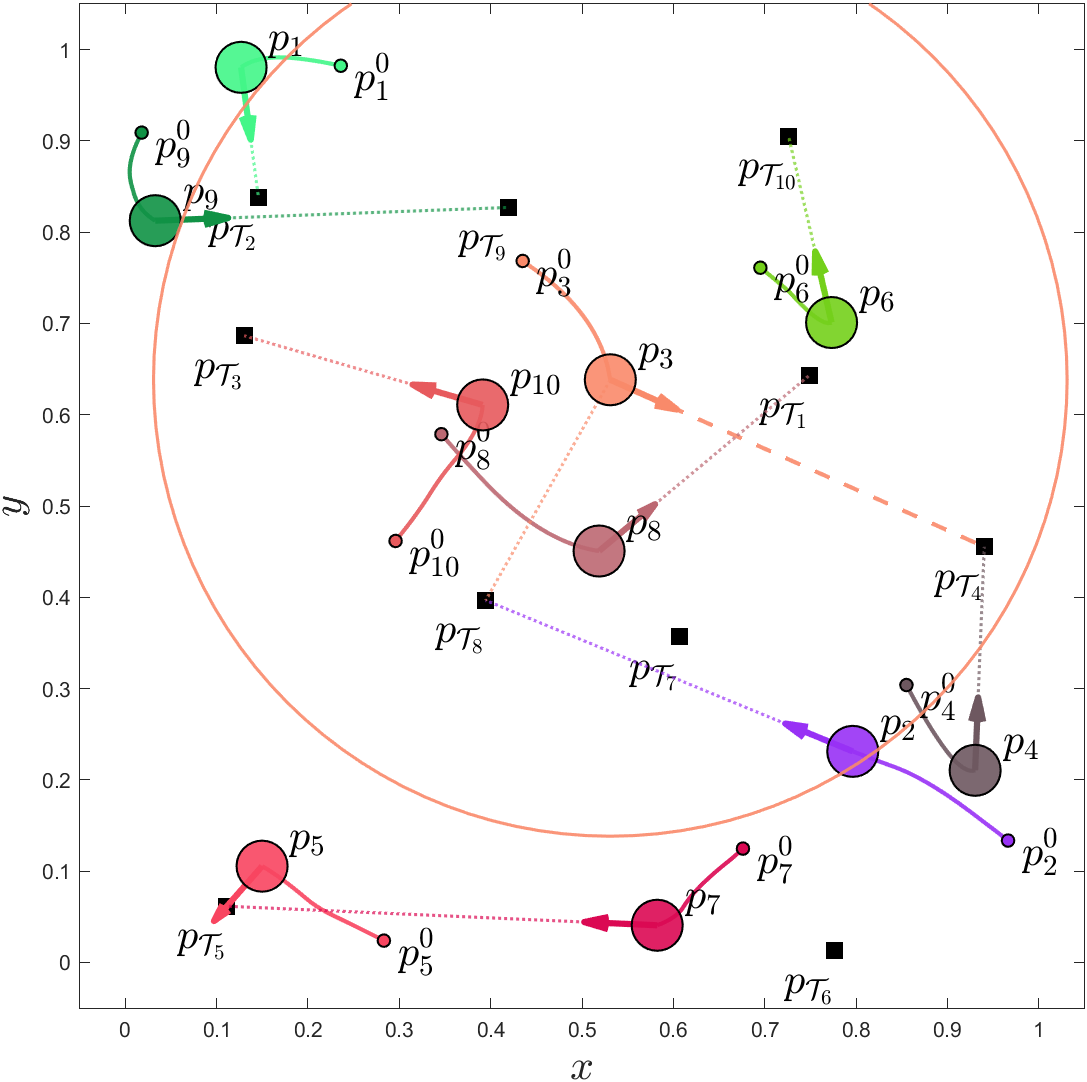}\label{fig:Ex_Im2}}
\subfigure[$t=10$]{\includegraphics[width=5.2cm]{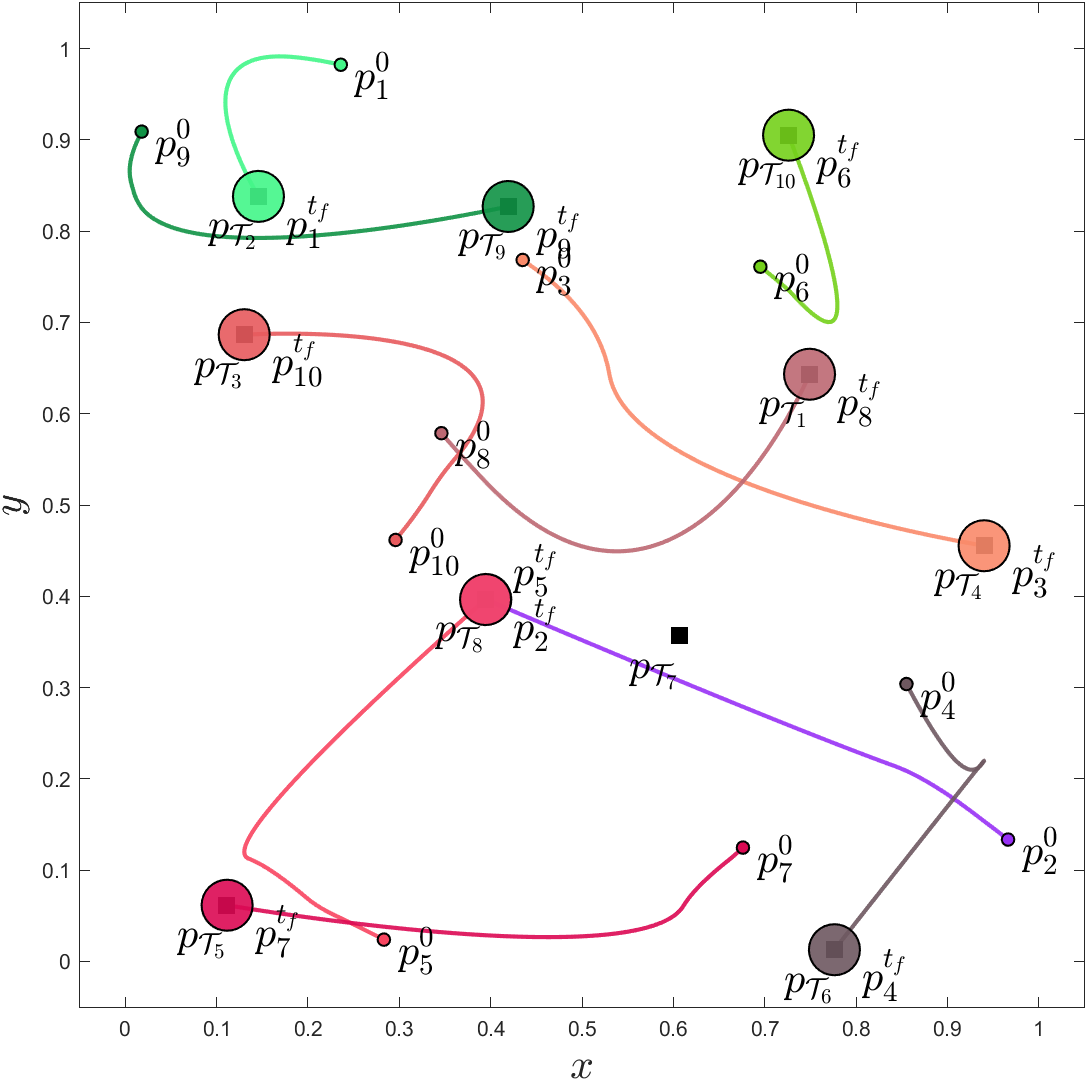}\label{fig:Ex_Im2}}
\caption{Dynamic task allocation for the range constrained case (GRM, $n=p=10$, $\varrho=0.5$, $t_\f=10$, $\cU=3.6154$)}
\label{fig:RHTA}
\end{figure}

\begin{figure}
\centering
\subfigure[$t=0^+$]{\includegraphics[width=5.2cm]{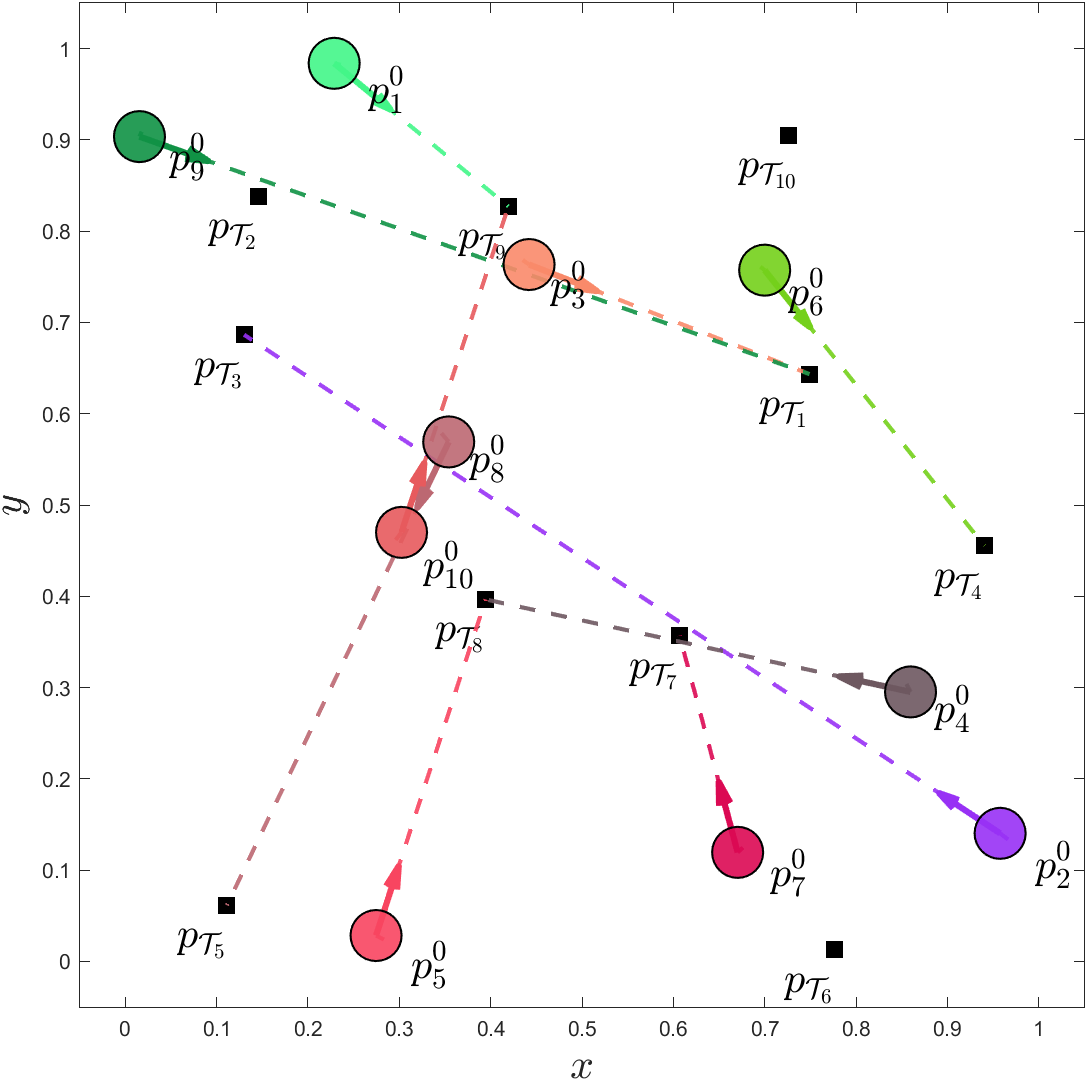}\label{fig:Ex_Im}}
\subfigure[$t=2.1$]{\includegraphics[width=5.2cm]{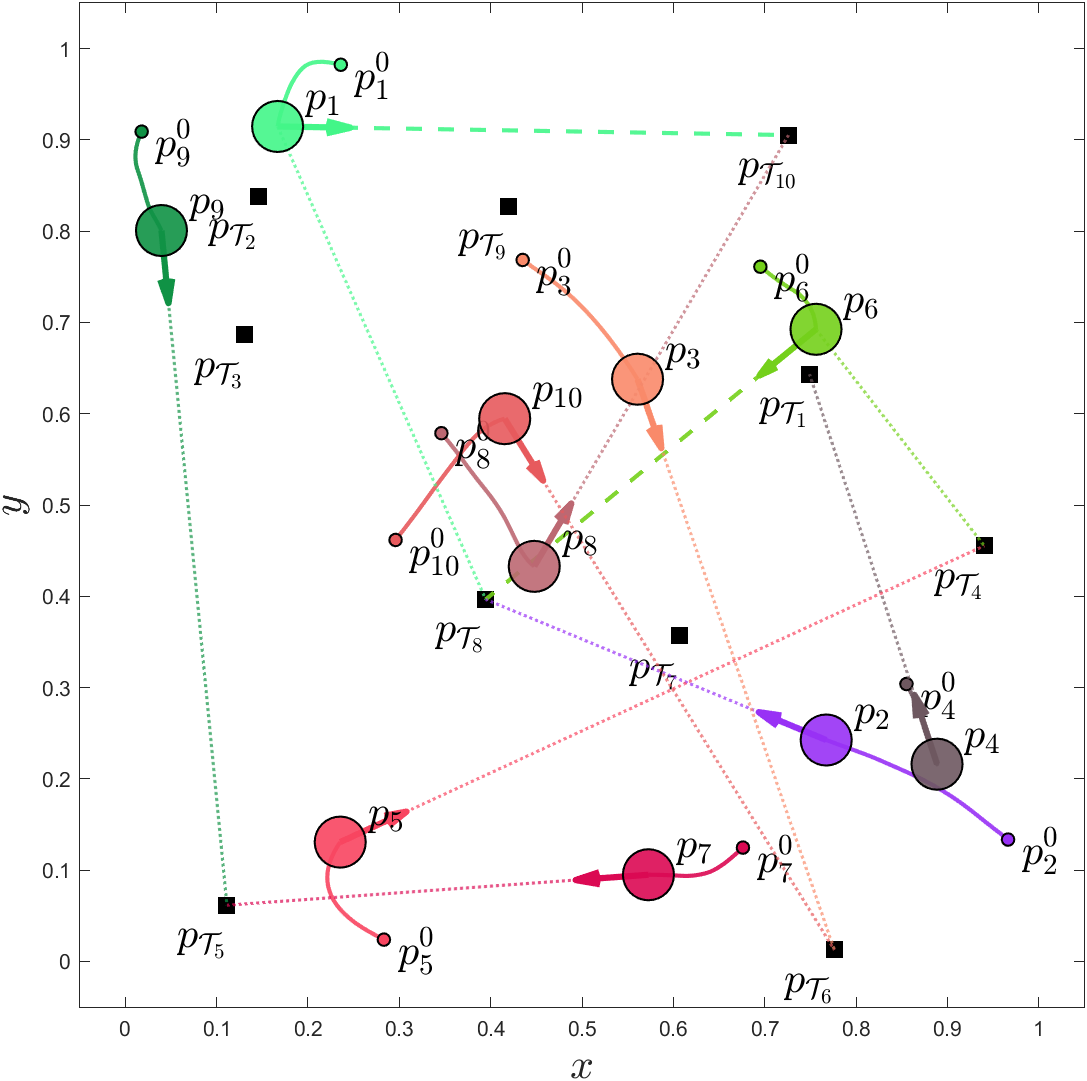}\label{fig:Ex_Im2}}
\subfigure[$t=10$]{\includegraphics[width=5.2cm]{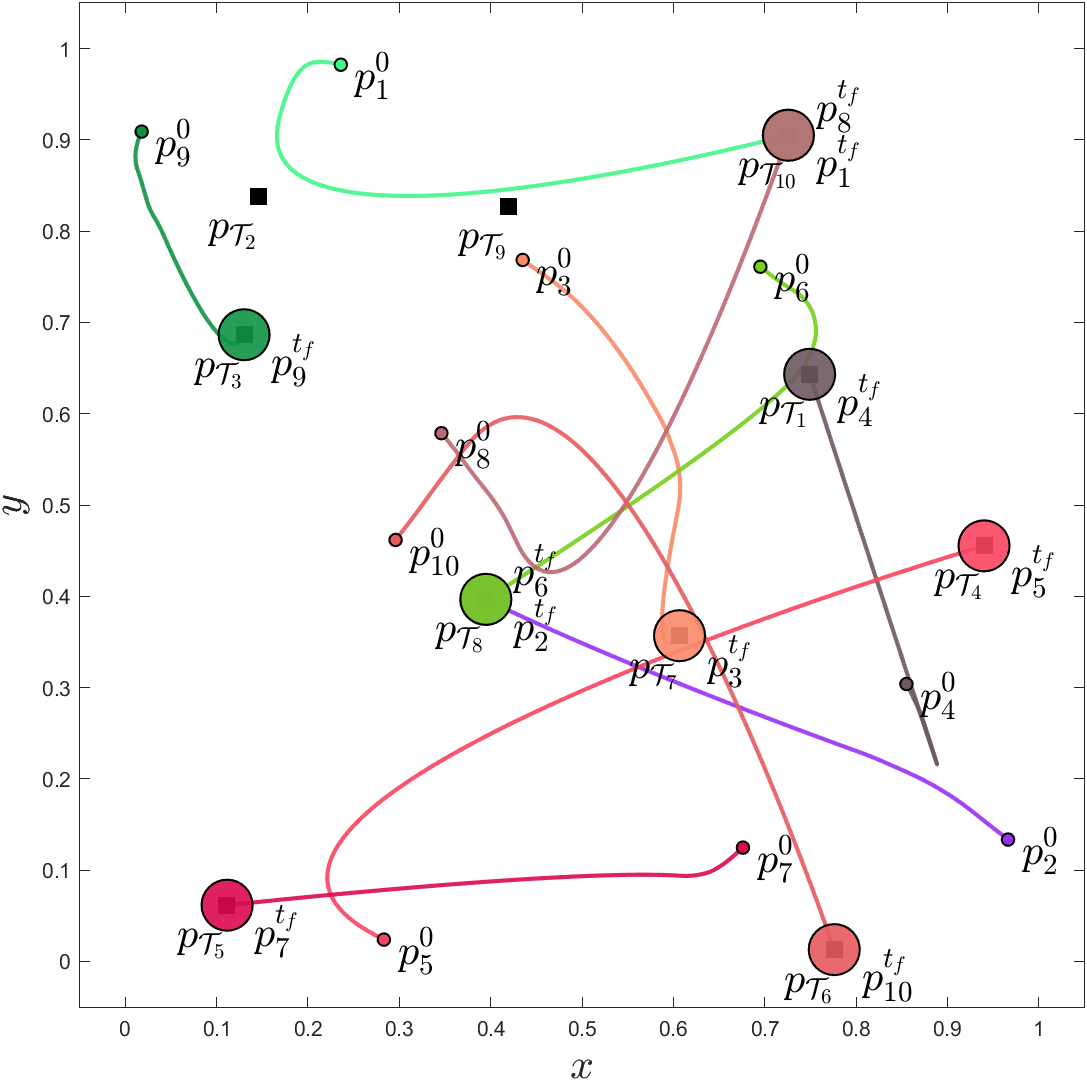}\label{fig:Ex_Im2}}
\caption{Dynamic task allocation for the range unconstrained case (GRM, $n=p=10$, $\varrho \rightarrow \infty$, $t_\f=10$, $\cU=4.0398$)}
\label{fig:RHTAUNC}
\end{figure}

\begin{figure}
\centering
\subfigure[$GRM$]{\includegraphics[width=4.2cm]{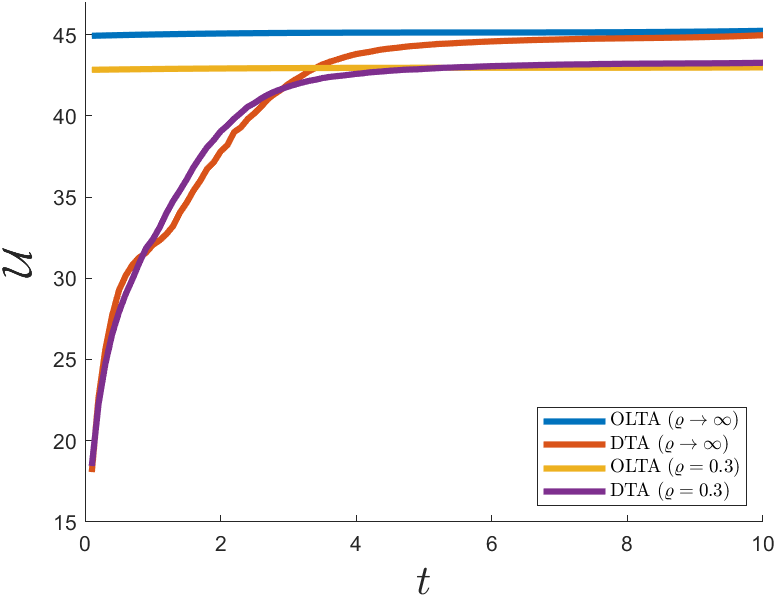}\label{fig:Ex_Im}}
\subfigure[$SAP$]{\includegraphics[width=4.2cm]{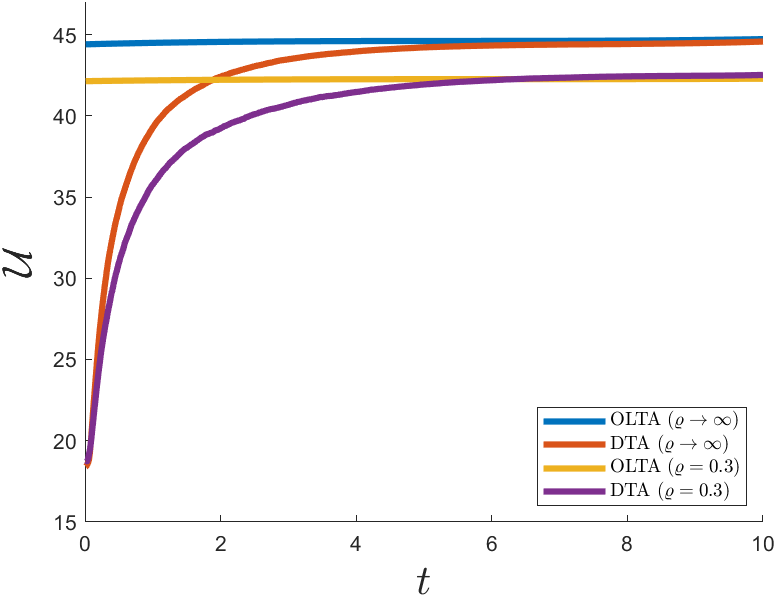}\label{fig:Ex_Im2}}
\caption{Team utilities versus time ($n=p=100$, $t_\f=10$)}
\label{fig:teamUtil}
\end{figure}

\begin{table*}[htbp]
\caption{Team Utilities ($n=p=100$) for range constrained and unconstrained cases}
\begin{center}
\begin{tabular}{|c|c|c|c|c|c|c|c|c|}
\hline
\multicolumn{1}{|c|}{}&\multicolumn{4}{c|}{GRM}&\multicolumn{4}{c|}{SAP} \\
\cline{1-9}
\multicolumn{1}{|c|}{}&\multicolumn{2}{c|}{OLTA}&\multicolumn{2}{c|}{DTA}&\multicolumn{2}{c|}{OLTA}&\multicolumn{2}{c|}{DTA} \\
\cline{1-9}
$t_\f$ & \textit{$\varrho \rightarrow \infty$}& \textit{$\varrho=0.3$}& \textit{$\varrho \rightarrow \infty$}& \textit{$\varrho=0.3$}& \textit{$\varrho \rightarrow \infty$}& \textit{$\varrho=0.3$}& \textit{$\varrho \rightarrow \infty$}& \textit{$\varrho=0.3$} \\
\hline
2& 37.1087& 39.2745& 25.4762& 33.9018& 40.4893& 39.3430& 35.8556& 36.2413 \\
\hline
5& 44.4502& 42.6770& 43.2366& 42.4138& 44.2700& 41.9399& 43.3736& 41.4702 \\
\hline
10& 45.2247& 42.9761& 44.9560& 43.2582& 44.7164& 42.2792& 44.5670& 42.5155 \\
\hline
\end{tabular}
\label{tabl:teamUtil}
\end{center}
\vspace{-2mm}
\end{table*}

We will present numerical simulations for both Problem~\ref{problemOLTA} (OLTA) and Problem~\ref{problemRHTA} (DTA) based on the SAP and GRM algorithms from~\cite{p:shamma2007} for both the range constrained and unconstrained cases. We will use a constant time step $\delta t$ (although in Proposition~2, we proposed a decreasing time step, it turns out that a sufficiently small constant step is adequate for our simulations). The negotiation process for the OLTA ran for $k$ rounds (all these rounds took place at time $t=0$ per the discussion in Section~\ref{ss:negotiate}) in order to converge to a Nash equilibrium before the agents start moving toward the states corresponding to their assigned tasks. We have used $k=100$ for the GRM algorithm and $k=1000$ for the SAP algorithm. The negotiation process for the DTA starts with a random task assignment profile at $t=0$ and subsequently, the agents continue to update their utilities and individual assignments at every time step ($k=t_\f/\delta t$). We have noticed that $\delta t$ must be smaller for the SAP algorithm than the GRM algorithm to achieve convergence. For this reason, we select $\delta t=0.1$ for GRM and $\delta t =0.01$ for SAP when solving the DTA problem. Per the discussion in Section~\ref{s:greedy}, the agents' utilities are not updated after time $t=t_\f-\epsilon$ whereas the negotiation process continues until $t=t_\f$. In our simulations we have used the following parameter values: $\epsilon = t_\f/20,\ p_i^0 \in [0,1]^2,\ \dot{p}_i^0 \in [-0.1,0.1]^2,\ p_{\cT_j} \in [0,1]^2,\ \bar{r}_{\cT_j} \in [0,1],\ \textrm{and}\ p_{i,j} \in [0,1]$ where $i \in [1,n]_d$ and $j \in [1,p]_d$. For the implementation of the GRM and SAP algorithms, we have used $\rho =0.1$ (discount factor), $\alpha = 0.5$ (parameter for the agents' willingness to optimize at each time step) and randomization level $\tau = 10/k^2$. Finally,  $\varrho \in \{0.3, 0.5\}$ (parameter for range constrained implementations of SAP and GRM). All the graphs and numerical outcomes presented herein are averaged data from $10^2$ simulation runs.

Figures \ref{fig:RHTA} and \ref{fig:RHTAUNC} illustrate the evolution of the agents trajectories computed for the DTA problem at different time instants for the range constrained and the unconstrained cases, respectively. In particular, the dash lines indicate the current task assignments whereas the solid curves correspond to the past segments of the agents' trajectories. Fig. \ref{fig:teamUtil} shows that both the team utility obtained by the GRM and SAP negotiation protocols for the DTA problem reach the team utility attained by the solution to the OLTA problem. In addition, the negotiations converge to a pure strategy Nash equilibrium as $t \rightarrow t_\f$ in agreement with  Proposition~\ref{prop:mainresult}. 
Table \ref{tabl:teamUtil} shows the values of the total team utility $\cU$ for different scenarios for both the range constrained and range unconstrained cases with a significant number of agents and tasks and for different values of the terminal time $t_\f$. We observe that for the DTA problem the team's performance improves as $t_\f$ increases. As we have discussed in Remark~3, when $t_\f$ is large, then the equilibrium assignment profile corresponds to a ``steady-state'' profile in which case the performance achieved by the solutions to both the OLTA and DTA problems are expected to be similar. The obtained results confirm the latter claim.

\section{Concluding Remarks}\label{s:concl}

In this paper, we have presented a framework to address dynamic task allocation problems for multi-agent systems with state-dependent utilities. Our approach, which leverage game-theoretic learning algorithms for the solution of static potential games, offers a practical solution to a class of more realistic and challenging dynamic task allocation problems for autonomous mobile agents. In our future work, we plan to extend the results presented herein to even more realistic task allocation problems including scenarios with  deadlines attached to tasks, pop-up tasks and agents with varying capabilities and preferences.

\bibliographystyle{ieeetr}
\bibliography{bakolas,gamesref}
\end{document}